\documentclass[11pt]{article}
\usepackage[utf8]{inputenc}
\usepackage{style}
\usepackage{framed}
\usepackage{nicefrac}
\usepackage[parfill]{parskip}
\usepackage{setspace}
\usepackage{hyperref}
\usepackage[ruled,vlined,linesnumbered]{algorithm2e}
\usepackage{algorithmicx}
\usepackage{thm-restate}
\usepackage{amsmath}

\let\oldnl\nl
\newcommand{\nonl}{\renewcommand{\nl}{\let\nl\oldnl}}

\newtheorem{SDP}[theorem]{SDP}

\author{Suprovat Ghoshal\\IISc\footnote{Indian Institute of Science, Bangalore, India.} 
\\suprovat@iisc.ac.in
\and Anand Louis\\IISc\footnotemark[1]\\anandl@iisc.ac.in
\and Rahul Raychaudhury\\IISc\footnotemark[1]\\rahulr@iisc.ac.in}

\date{}
\title{Approximation Algorithms for Partially Colorable Graphs}
\begin{document}
\begin{titlepage}
\maketitle
\begin{abstract}
Graph coloring problems are a central topic of study in the theory of algorithms.
We study the problem of partially coloring {\em partially colorable graphs}.
For $\alpha \leq 1$ and $k \in \mathbb{Z}^+$, 	
we say that a graph $G=(V,E)$ is $\alpha$-partially $k$-colorable, if there 
exists a subset $S\subset V$ of cardinality $\lvert S \rvert \geq \alpha \lvert V \rvert$ such that
the graph induced on $S$ is $k$-colorable.
Partial $k$-colorability is a more robust structural property of a graph than $k$-colorability. 
For graphs that arise in practice, partial $k$-colorability might be 
a better notion to use than $k$-colorability, since data arising in practice often contains 
various forms of noise.
	
We give a polynomial time algorithm that takes as input a $(1 - \epsilon)$-partially $3$-colorable 
graph $G$ and a constant $\gamma \in [\epsilon, 1/10]$, and colors a $(1 - \epsilon/\gamma)$ fraction of the vertices
using $\tilde{O}\left(n^{0.25 + O(\gamma^{1/2})} \right)$ colors.
We also study natural semi-random families of instances of partially $3$-colorable graphs and 
partially $2$-colorable graphs, and give stronger bi-criteria approximation guarantees for these 
family of instances.

\end{abstract}

\end{titlepage}

\newcommand{\paren}[1]{\left(#1\right)}
\newcommand{\pkc}[1]{\ensuremath{{\sf P}#1{\sf C}} }
\newcommand{\pkcr}[1]{\ensuremath{{\sf P}#1{\sf C}^{\cal R}\paren{n,p}} }
\newcommand{\pthrc}{\pkc{3} }
\newcommand{\pthrcr}{\pkcr{3} }
\newcommand{\Abs}[1]{\left\lvert#1\right\rvert }
\newcommand{\defeq}{\stackrel{\textup{def}}{=} } 

\section{Introduction}

Graph coloring problems are a central topic of study in the theory of algorithms
\cite{wigderson_1983,KMS98,AG11,kawarabayashi_thorup_2017}.
An undirected graph $G=(V,E)$ is said to be $k$-colorable if there exists an assignment of colors 
$f:V \to [k]$ such that $f(u) \neq f(v)$ for each $\set{u,v} \in E$.
For a graph $G$, the minimum value of $k$ for which it is $k$-colorable is called its chromatic number.
Computing a $3$-coloring of a $3$-colorable graph is a fundamental NP-hard problem.
Efficiently computing a coloring of a $3$-colorable graph which only uses a few colors is a major
open problem in the study of algorithms.
The current best known algorithm colors a $3$-colorable graph on $n$ vertices using $O(n^{0.199})$ colors \cite{kawarabayashi_thorup_2017}.
We study the problem of coloring partially colorable graphs. 

\begin{definition}
\label{def:pkc}
An undirected graph $G = (V,E)$ is defined to be $\alpha$-partially $k$-colorable, denoted by
$\alpha$-\pkc{k}, if there
exists a subset $V_{\rm good} \subset V$ such that $\Abs{V_{\rm good}} \geq \alpha \Abs{V}$ and the graph 
induced on $V_{\rm good}$ is $k$-colorable. We will call such a set $V_{\rm good}$ the set of {\em good} vertices,
and $V_{\rm bad} \defeq V \setminus V_{\rm good}$ the set of {\em bad} vertices. 
\end{definition}

We remark that for a given graph the partitioning of the vertex set $V$ into $V_{\rm good}$ and $V_{\rm bad}$ may not be unique. In such cases, the claims we make in this paper will hold for any such fixed partition. 

It is well known that for a fixed $k$, the problem of determining whether a given graph is $k$-colorable 
is an NP-hard problem \cite{Karp72}. Therefore, determining whether a graph belongs to $1$-\pkc{k}
is an NP-hard problem, and hence, computing the largest value of $\alpha$ for which a graph
belongs to $\alpha$-\pkc{k} is also an NP-hard problem.

Note that a graph that is $(1-\epsilon)$-partially $3$-colorable can have chromatic number 
as large as $\Abs{V_{\rm bad}} = \epsilon n$. Therefore, the notion of the chromatic number of the graph does not 
capture the structural property ($3$-colorability) satisfied by most of the graph.
Partial $k$-colorability is a more robust stuctural property than $k$-colorability.
Therefore, for graphs that arise in practice, partial $k$-colorability might be a better notion to
to use than $k$-colorability, since data arising in practice often contains various forms of noise;
the notion of bad vertices can be used to capture some types of noisy vertices in the graph.  

\paragraph{Other notions of partial $k$-coloring.}Another related notion 
of partial coloring is the following.
\begin{definition}
\label{def:pkc-edge}
An undirected graph $G = (V,E)$ is defined to be $\alpha$-partially $k$-colorable, 
if there exists a coloring of the vertices $f : V \to [k]$ such that for at least 
$\alpha \Abs{E}$ edges $\set{u,v}$, $f(u) \neq f(v)$.
\end{definition}
This definition, which asks that the coloring should ``satisfy'' at least $\alpha$ 
fraction of the edges, can be viewed as the {\em edge} version of partial $k$-colorability,
whereas Definition \ref{def:pkc} can be viewed as the {\em vertex} version of partial $k$-colorability.
For a fixed constant $k$, computing the maximum value of $\alpha$ for which the input graph satisfies 
Definition \ref{def:pkc-edge} can be formulated as a Max-$2$-{\sf CSP} with alphabet size $k$; approximation algorithms
for Max-$2$-{CSP}s have been extensively studied in the literature \cite{R08,RS09a,BRS11} etc.
Therefore, we focus our attention on Definition \ref{def:pkc}.

\subsection{Our Results}

We give an efficient (bi-criteria) approximation algorithm for coloring partially $3$-colorable graphs.
\begin{restatable}{thm}{ColMain}  \label{thm:3col-main}
There exists a polynomial time algorithm that takes as input 
a $(1-\epsilon)$-\pthrc~graph $G = (V,E)$ and any fixed choice of $\gamma \in [\epsilon, 1/100]$, 
and produces a set $S \subset V$ such that $\Abs{S} \leq(3\epsilon/\gamma) \Abs{V}$
and a coloring of $V \setminus S$ using $\tilde{O}(n^{0.25 + O(\gamma^{1/2})})$ colors\footnote{
	$\tilde{O}(\cdot)$ hides factors polylogarithmic in $n$.}.
\end{restatable}
We point out that the above theorem gives a bi-criteria approximation guarantee which exhibits the tradeoff between the size of the set $S$, and the number of colors used to color the remaining graph $G[V\setminus S]$. In particular, setting $\gamma = \sqrt{\epsilon}$ in the above theorem gives us the following guarantee. Given a $(1-\epsilon)$-\pthrc graph, one can color $(1 - \sqrt{\epsilon})$-fraction of its vertices using $\tilde{O}(n^{0.25 + \epsilon^{1/4}})$-colors. Using similar techniques we can give an efficient approximation algorithm for the partial $2$-coloring setting as well. For completeness, we formally state the result below{\footnote{We implicitly use the algorithm in the degree reduction step of the algorithm from Theorem \ref{thm:3col-main}. See Claim \ref{cl:step(ii)} for details.} :

\begin{restatable}{prop}{TwoColGen} \label{prop:2col-gen}
		There exists a polynomial time algorithm that takes as input 
		a $(1-\epsilon)$-\pkc{2} graph $G = (V,E)$ and any fixed choice of $\gamma \in [\epsilon, 1/100]$, 
		and produces a set $S \subset V$ such that $\Abs{S} \leq(4\epsilon/\gamma) \Abs{V}$
		and a coloring of $V \setminus S$ using $\tilde{O}(n^{2\gamma})$ colors.
\end{restatable}

We also study a semi-random family of partially colorable graphs $\alpha$-\pkcr{k}, which we define as follows.
\begin{definition}
\label{def:pkcr}
An instance of $\alpha$-\pkcr{k} is generated as follows.  
\begin{enumerate}
	\item Let $V$ be a set of $n$ vertices. Arbitrarily partition $V$ into sets $V_{\rm good}$ and $V_{\rm bad}$
		such that $\Abs{V_{\rm good}} \geq \alpha n$.
	\item Add edges between an arbitrary number of arbitrarily chosen pairs of vertices in 
		$V_{\rm good}$ such that the graph induced on $V_{\rm good}$ is $k$-colorable. 
	\item Add edges between an arbitrary number of arbitrarily chosen pairs of vertices in $V_{\rm bad}$. 
	\item Between each pair of vertices in $V_{\rm good} \times V_{\rm bad}$, independently add an edge with probability
		$p$. We call this set of edges $E_0$.
	\item Add arbitrary number of edges between pairs of vertices of $V_{\rm good} \times V_{\rm bad}$. We call this set of edges $E_1$. 	
\end{enumerate}
Output the resulting graph. 
\end{definition}

In the study of approximation algorithms for NP-hard problems, there have been many works
studying algorithms random and semi-random instances of various problems \cite{BS95,FK01,KMM11,MMV12,MMV14}.  
Random and semi-random instances are often good models
for instances arising in practice; designing algorithms specifically for such instances,
whose performance guarantee is significantly better than guarantees for general instances,
could have more applications in practice. 
Moreover, from a theoretical perspective, designing algorithms for semi-random instances
helps us to better understand what aspects of a problem make it intractable. 
We study our semi-random model $\alpha$-\pkcr{k} for the same reasons. 
The following is our main result. 

\begin{restatable}{thm}{ColRandom} \label{thm:3col-random}
Suppose there exists an efficient algorithm which colors a $3$-colorable graph using $n^{\theta}$ colors. Then the following holds for all choices of $\epsilon = \Omega(\log n/n)$ and $p \ge (\epsilon\theta^{-2})^{O(\theta)}$. There exists a polynomial time algorithm that takes as input a graph $G$ sampled from
$(1-\epsilon)$-\pthrcr  and produces a set $S$ such that $\Abs{S} = O\paren{ \epsilon\theta^{-2} n p^{-(O(1/\theta))}}$
and a coloring of $V \setminus S$ using at most $n^{\theta}$ colors with high probability.
Moreover, the algorithm runs in time $n^{O(1/\theta)}{\rm poly}(n)$. 
\end{restatable}

In particular, instantiating the above theorem with the algorithm from \cite{kawarabayashi_thorup_2017}, w.h.p., we can colors $(1- O(\epsilon))n$ fraction of vertices with $\tilde{O}(n^{0.199})$-colors. We also study the partial $2$-coloring problem in the semi-random setting. Our guarantees for this setting are as follows:

\begin{restatable}{thm}{TwoColRandom} \label{thm:2col-random}
		Let $\epsilon = \Omega(\log n/n)$ and $p > \sqrt{\epsilon}$. Then, there exists a polynomial time algorithm that takes as input a graph $G$ sampled from
		$(1-\epsilon)$-\pkcr{2}, and with high probability, produces a set $S \subseteq V$ such that $\Abs{S} = O\paren{ \epsilon n p^{-2}}$ and the induced subgraph on the remaining vertices $G[V\setminus S]$ is $2$-colorable. 
\end{restatable}
In particular, in the above theorem the number of vertices removed is bounded by $O(\epsilon n)$ which is stronger than the best known bound of $O(\sqrt{\log n}.\epsilon n)$~\cite{ACMM05} in the adversarial setting.

\subsection{Related Work}
\label{sec:related}
\paragraph{$3$-colorable graphs.} 
There is extensive literature on algorithms for coloring $3$-colorable graphs. Wigderson \cite{wigderson_1983} gave a simple combinatorial algorithm that used $O(n^{\frac{1}{2}})$ colors. Blum \cite{Blum94} improved the number of colors used to $\tilde O(n^\frac{3}{8})$. These algorithms used purely combinatorial techniques. Karger, Motwani and Sudan \cite{KMS98} used semidefinite programming to develop an algorithm, which when balanced with Wigderson's technique \cite{wigderson_1983} used $\tilde{O}({n}^\frac{1}{4})$ colors. Blum and Karger \cite{blum_karger_1997} improved the number of colors used to $\tilde{O}({n}^\frac{3}{14})$ by combining the techniques used in \cite{Blum94} and \cite{KMS98}. Arora, Chlamtac and Charikar \cite{AC06} got the bound down to $\tilde{O}({\Delta}^{0.21111})$ using techniques from the ARV algorithm \cite{arora2009expander}, which was further improved by Chlamtac \cite{chlamtac_2007} to $\tilde{O}({n}^{0.2072})$ using SDP hierarchies. Using new combinatorial techniques, Kawarabayashi and Thorup improved the approximation bound to $\tilde{O}({n}^{0.2049})$ in \cite{kawarabayashi_thorup_2012}. Subsequently, by combining their techniques with \cite{chlamtac_2007}, they were able to give a approximation of $\tilde{O}({n}^{0.19996})$ \cite{kawarabayashi_thorup_2017}, which is the current state of the art.

\paragraph{Partially $2$-colorable graphs.}
The partial $2$-coloring problem, better known as Odd Cycle Transversal (OCT) in the literature, has also been studied extensively. Formally, the setting here is as follows. We are given a $(1-\epsilon)$-partially $2$-colorable graph $G = (V,E)$ and the objective is to find a set $S$ of minimum size such that $G[V\setminus S]$ is $2$-colorable (i.e., odd cycle free). Yannakakis first showed that it is NP-Complete in \cite{Yann78}. Later, Khot and Bansal~\cite{KB09} showed that OCT is hard to approximate to any constant factor, assuming the Unique Games Conjecture. From the algorithmic side,  via a reduction through the Min2CNF Deletion problem, \cite{GVY96} gave a $O(\log n)$ approximation for the problem. This was later improved to $O(\sqrt{\log n})$ by \cite{ACMM05} by using techniques from the Arora-Rao-Vazirani~\cite{arora2009expander} algorithm for sparsest cut. This problem has also been studied under the lens of parameterized complexity. In \cite{RVS04}, Reed et al. showed that OCT is fixed parameter tractable when parameterized by the number of bad vertices, following which a sequence of works \cite{KB09}\cite{NVS12}\cite{LNVS14} gave algorithms with improved running times.  

\paragraph{Partially $3$-colorable graphs}

In contrast to the $3$-colorable setting, there has been very little work on coloring partially $3$-colorable graph. The paper which is closest to our setting is by Kumar, Louis and Tulsiani~\cite{kumar_louis_tulsiani}, which also addresses the partial $3$-coloring problem, albeit in a more restrictive setting.  Assuming that the $(1-\epsilon)$-partially $3$-colorable graph has threshold rank $r$ and the $3$-coloring on the good vertices satisfies certain psuedorandomness properties, they give an algorithm which $3$-colors $1-O(\gamma + \epsilon)$ fraction of vertices in time $(r.n)^{O(r)}$. 

\paragraph{Graph problems in Semi-random Models}

The semi-random model used in this paper is similar to semi-random models which have been considered for the Max-Independent Set problem~\cite{BS95}~\cite{FK01}~\cite{stein17}~\cite{MMT18}. Semi-random models offer a natural way of understanding the complexity of problems in settings which are less restrictive than worst case complexity, but are still far from being average case. While semi-random models were first introduced for studying graph coloring in ~\cite{BS95}, it has also subsequently been used to study several other fundamental problems such as Unique Games~\cite{KMM11}, Graph Partitioning~\cite{MMV12}, Clustering~\cite{MMV14}, to name a few.The problem of coloring $3$-colorable graphs has also been studied in average-case and planted models. Alon and Kahale \cite{AK97} gave an efficient algorithm that finds an exact $3$-Coloring of a random $3$-Colorable graph with high probability. David and Fiege~\cite{FR16} studied the complexity of finding a planted random/adversarial $3$-coloring for both adversarial and random host graphs.

\subsection{Discussion and Proof Overview}
\label{sec:overview}

{\bf Adversarial Model}: The key component in most approximation algorithms for $3$-coloring involves solving a SDP relaxation of the $3$-coloring problem, and followed by a randomized rounding procedure for coloring the graph. The standard SDP relaxation for $3$-coloring is the following which was introduced in \cite{KMS98}:
\begin{SDP}[Exact $3$-Coloring SDP]
\label{SDP:kms}
\begin{eqnarray*}
	\text{minimize}  &  0 & \\
	\text{subject\space to}&  v_i\cdot v_j \le -\frac{1}{2} & \quad \forall \{i,j\}\in E \\
	& \|v_i\|^2=1 &\forall i \in V
\end{eqnarray*}\\
\end{SDP}

SDP \ref{SDP:kms} doesn't optimize any objective function, 
it finds a feasible solution which satisfies all the constraints of SDP. The intended solution to the above SDP is as follows. Let $\sigma:V \to \set{1,2,3}$ be any legal coloring of $G$. Furthermore, let $u_1,u_2,u_3 \in \mathbbm{R}^2$ be any three unit vectors satisfying $\langle u_i , u_j \rangle = -1/2$ for every $i,j \in \{1,2,3\}, i \neq j$. We identify the vector $u_i$ with the color $i$, and assign $v_j = u_{\sigma(j)}$ for every $j \in V$. It can be easily verified that this is a feasible solution to the above SDP. As is usual, while the SDP in general may not return the above vector coloring, one can round a feasible vector coloring to color the graph using not too many colors \cite{KMS98}. 
The approximation guarantee is usually of the form $\Delta^c$ (for some $c \in (0,1)$), where $\Delta$ is the maximum degree of the graph. 

Since in general, one cannot hope to have a degree bound on the graph, the above step is usually preceded by a {\em degree reduction} sub-routine. Note that if a graph is $3$-colorable (more generally $k$-colorable), then the graph induced on the neighbours of any vertex $v$ is $2$-colorable (more generally $k-1$ colorable). Since a $2$-colorable graph can be colored with $2$ colors efficiently, the graph induced on any vertex and its neighbours can be colored efficiently with $3$ colors. 
Therefore, fixing a threshold $\Delta$, this procedure iteratively removes vertices (and their neighbours) having degree larger than $\Delta$ from the graph while coloring them with few colors, and terminates when maximum degree of the remaining graph is at most $\Delta$.
In particular, if the degree reduction step uses $f(n,\Delta)$ colors, then the total number of colors used by the algorithm is at most $f(n,\Delta) + \Delta^c$. Then one can optimize the choice of $\Delta$ for giving the best possible approximation guarantee. This degree reduction approach and its variants, first studied by Wigderson \cite{wigderson_1983}, has been subsequently used in almost all known approximation algorithms for graph coloring 

In translating the above template to the setting of partially $3$-colorable graphs, we face several
immediate challenges. SDP \ref{SDP:kms} is guaranteed to return a feasible solution only for $3$-colorable graphs, 
it might be infeasible if the graph is not $3$-colorable.
If we could compute the set of good vertices then we could use
SDP \ref{SDP:kms} only on the set of good vertices. 
However, in general, the problem of identifying the set of good vertices is NP-hard (Fact \ref{fact1}).
Finally, the preprocessing steps for degree reduction rely heavily on the combinatorial 
structural properties of the neighborhood of vertices in {\em exactly} $3$-colorable graphs, 
which, in general, may not be satisfied by  {\em partially} $3$-colorable graphs.

Our approach is to begin with an SDP relaxation that tries to solve both problems together: identifying
the set of bad vertices, and coloring the set of good vertices. 
We introduce variables $w_1,w_2,\ldots,w_n$ where the $i^{th}$ variable $w_i$ is 
meant to indicate if vertex $i$ is bad. Additionally, for every edge $(i,j) \in E$, we introduce
slack variables $z_{ij}$ which are meant to indicate if at least one of the vertices $i,j$ is bad. 
Using the slack variables we relax the edge constraints as $\langle v_i, v_j \rangle \le -1/2 + (3/2)z_{ij}$. 
Finally, we connect the edge indicator variables with vertex indicator variables using constraints 
of the form $z_{ij} \le w_i + w_j$. 
Since we want the 
set of bad vertices to be small, our objective function will be to minimize $\sum_{i \in V} w_i$.
Our SDP relaxation is the following.  

\begin{SDP}[Partial $3$-Coloring SDP]
\label{SDP:p3c}
\begin{eqnarray*}
	\text{minimize}  &  \sum_{i\in V} {w_i} & \\
	\text{subject\space to}&  \langle v_i,  v_j \rangle \le -\frac{1}{2}+\frac{3}{2}z_{ij} & \quad \forall \{i,j\}\in E \\
	& z_{ij} \leq w_i+ w_j & \forall \{i,j\}\in E \\
	& 0 \leq z_{ij} \leq 1 & \forall \{i,j\}\in E \\
	& 0 \leq w_{i} \leq 1 &\forall i\in V \\
	& \|v_i\|^2=1 &\forall i \in V
\end{eqnarray*}\\
\end{SDP}

Since the optimal ``integer solution'' forms a feasible solution to the SDP relaxation, 
it is easy to show that for a $(1-\epsilon)$-partially $3$-colorable graph, the optimal of 
the above SDP is at most $\epsilon n$. Therefore by Markov's inequality, we get that for a 
large fraction of $i \in [n]$, the $w_i$ variables are small. Let $V' \subset V$ be the set 
of vertices with small $w_i$. Since $\Abs{V\setminus V'} = O(\epsilon n)$, we can focus on 
coloring the induced subgraph $G' = G[V']$. $G'$ has the following nice property: 
{\em for every edge $(i,j)$ in G', the corresponding edge constraint is approximately satisfied 
i.e., $\langle v_i , v_j \rangle \le -1/2 + o_\epsilon(1)$, where the second term goes to $0$ as 
$\epsilon$ goes to $0$}. We call such graphs as being approximately vector $3$-colorable (See Definition \ref{defn:appx-vect-col} for a formal description). 
We use this property crucially in designing our preprocessing step. 

We observe that the neighborhood of any vertex in an approximately vector $3$-colorable graph is 
approximately vector $2$-colorable. Furthermore, we show that approximately vector $2$-colorable graphs 
are {\em short odd cycle} free. Graphs having this property are known to have large independent sets 
which can be found efficiently \cite{RS85ramsey}. Thus one can find such large independent sets recursively to 
color the neighborhood of large degree vertices using a small number of colors. 

For the randomized rounding step, we observe that hyperplane rounding based procedures are naturally robust to small perturbations, and the arguments for analyzing the guarantees of such procedures hold even when the edge constraints are approximately satisfied. In particular, we can use known randomized rounding algorithm as is, while adapting the analysis to account for the edge constraints being satisfied approximately.

{\bf Semi-random model}: While the guarantees of our algorithm from the adversarial setting also apply to the semi-random instances, here we seek to achieve the best known approximation bounds for exactly $3$-colorable graphs. We begin by describing two distinct classes of instances which illustrate the technical challenges in designing such an algorithm.

In this setting, the adversary is free to choose $G[V_{\rm bad}]$ in a way such that it is noisy and has large chromatic number (e.g, graphs sampled from Erdos Renyi random model). For such instances, it is easy to see that the only way an algorithm can have good approximation guarantees is when it can eliminate a significant fraction of from $V_{\rm bad}$.  Then, for a start, one can hope to address this setting by first using a preprocessing step that deletes $V_{\rm bad}$ and then running the best possible approximation algorithm on the graph induced on the remaining vertices.

On the other hand, the adversary can also choose $G[V_{\rm bad}]$ in a way so that it is {\em structurally indistinguishable} from the good subgraph $G[V_{\rm good}]$. For instance, suppose the good subgraph $G[V_{\rm good}]$ is a randomly sampled unbalanced bipartite graph, where the smaller side (which we call $V_S$) has size at most $\epsilon n$. Then the adversary can choose $V_{\rm bad}$ to be an independent set, in which case the entire graph is $3$-colorable. In particular, it is information theoretically impossible to distinguish the set $V_{S}$ from $V_{\rm bad}$, since they are both independent sets and the edges incident on them are identically distributed. While the instances constructed here make it difficult to identify $V_{\rm good}$, they are also naturally easy instances for us. In particular, these instances are also $(1-\epsilon)$-partially $2$-colorable, and one can use tools for coloring partially $2$-colorable graphs to color these instances with small number of colors. 

However, the two cases above clearly do not cover the full range of instances that we can encounter in our model. Therefore, we need a way to relax the above two characterizations which allows for a seamless transition from one class of instances to other. It turns out that we can robustly characterize both classes of instances by the number of {\em vertex disjoint short odd cycles} present in the graph. Informally, if the number of short odd cycles is large, then with high probability, they will show up in the neighborhood of the bad vertices, and therefore this can be used to identify and eliminate $V_{\rm bad}$. We can then simply run the best known approximation algorithm on the remaining induced graph $G[V_{\rm good}]$. On the other hand, if the number of short odd cycles is small, by eliminating a small fraction of vertices, we can make the graph short odd cycle free. Finally, as discussed in the adversarial model setting, such graphs can be colored efficiently using a small number of colors by recursively finding large independent sets~\cite{RS85ramsey}.

\section{Preliminaries}

We introduce some notation used frequently in this paper. Throughout the paper, for a $(1-\epsilon)$-partially $3$-colorable graph $G = (V,E)$, we will write $V = V_{\rm good} \uplus V_{\rm bad}$ where $V_{\rm good}$ and $V_{\rm bad}$ are the set of good vertices and bad vertices as defined in Definition \ref{def:pkc}. For a subset $V' \subseteq V$, we use $G[V']$ to denote the subgraph induced on the set of vertices $V'$. For a subgraph $G' \subseteq G$, we shall use ${\rm vert}(G')$ to denote the vertex set of $G'$. Additionally, for any vertex $i \in {\rm vert}(G')$, we use $N_{G'}(i)$ denote the set of neighbors of $i$ in the graph $G'$. We use $\mathbbm{1}(\cdot)$ to denote the indicator function, and $\tilde{O}(\cdot)$ to hide terms which are polylogarithmic in the number of vertices. 
\paragraph{Approximate Vector Coloring.}
We begin by recalling the notion of vector coloring of a graph which was introduced in \cite{KMS98}.
\begin{definition}[Vector Coloring]					
\label{defn:vect-col}
Given a positive integer $k\in \mathbbm{N}$, we say that a graph $G = (V,E)$ is $k$-vector colorable 
if there exists unit vectors $v_1,v_2,\ldots,v_n \in \mathbbm{R}^d$ for some $d \in \mathbbm{N}$ which satisfy
\begin{equation*}
	\langle v_i , v_j \rangle \leq -\frac{1}{k-1}  \qquad \qquad \forall \set{i,j} \in E.
\end{equation*}
	
\end{definition}

We will use the notion of {\em approximate vector colorings} of a graph, which 
we define as follows. 

\begin{definition}[Approximate Vector Coloring]			
\label{defn:appx-vect-col}
Given a positive integer $k\in \mathbbm{N}$ and a $\gamma > 0$, we say that a graph $G = (V,E)$ is 
$(k,\gamma)$-vector colorable if there exists unit vectors $v_1,v_2,\ldots,v_n \in \mathbbm{R}^d$ 
for some $d \in \mathbbm{N}$ which satisfy
\begin{equation*}
	 \langle v_i , v_j \rangle \le -\frac{1}{k-1} + \gamma \qquad \qquad \forall \set{i,j} \in E.
\end{equation*}
\end{definition}
Observe that a graph that $(k,0)$ vector colorable is vector-$k$-colorable. We now state a couple of lemmas which illustrate some useful properties of approximate vector colorings. In \cite{KMS98}, it was observed that the vector chromatic number of sub-graph induced on the neighborhood of a vertex is strictly less than the vector chromatic number of the actual graph. In the following lemma, we observe that this property can be extended to approximate vector colorings as well.

\begin{lemma}					\label{lem:appx-2-col(a)}
	Let $G = (V,E)$ be $(3,\gamma)$-vector colorable, for some $0 < \gamma <1/10$. Then for any vertex $i \in V$, the graph induced on $N(i)$ is $(2,4\gamma)$-vector colorable.
\end{lemma}
\begin{proof}
		The proof of this lemma follows along the lines of Lemma 4.3 from \cite{KMS98}, which says that subgraphs induced by neighborhoods of vertices in vector $3$-colorable graphs are vector $2$-colorable. Without loss of generality, let $N_G(i) = \{1,2,\ldots,r\}$ and let $\{v_1,v_2,\ldots,v_r\}$ be the set of vectors which are a $(3,\gamma)$-vector coloring of $N_G(i)$. For every $j \in [r]$, we can write $v_j = v^{\|}_j + v^\perp_j$ where $v^{\|}_j$ and $v^\perp_j$ are the projections of $v_j$ along $v_i$ and $({\rm span}(v_i))^\perp$ respectively. Finally, for every $j \in [r]$ we define $\tilde{v}_j := v^\perp_j/\|v^\perp_j\|$ to be unit vector given by the projection of $v_j$ on the subspace $({\rm span}(v_i))^\perp$. It can be easily verified that $\tilde{v}_1,\tilde{v}_2,\ldots,\tilde{v}_r$ is a $(2,4\gamma)$-vector coloring of the graph induced on $N(v)$. To see this, fix any $j \in V$. By construction, we have $\|v^{\|}_j\| = |\langle v_i , v_j \rangle | \geq \frac12 - \gamma$, and therefore $\|v^\perp_j\| = \sqrt{1 - \|v^{\|}_j\|^2} \leq \sqrt{\frac34 + \gamma  - \gamma^2}$. Therefore for any $j,j' \in [r]$ such that $(j,{j'}) \in E$, using the orthonormal decomposition of $v_j$ and $v_{j'}$ we have 
		\begin{eqnarray*}
			\langle \tilde{v}_j , \tilde{v}_{j'} \rangle = \left\langle \frac{v^\perp_j}{\|v^\perp_j\|}, \frac{v^\perp_{j'}}{\|v^\perp_{j'}\|} \right\rangle 
			&=& \frac{1}{\|v^\perp_j \| \|v^\perp_{j'}\|}\Big(\langle v_j, v_{j'} \rangle - \langle v^{\|}_{j} , v^{\|}_{j'} \rangle\Big) \\
			&=& \frac{1}{\|v^\perp_j \| \|v^\perp_{j'}\|}\Big(\langle v_j, v_{j'} \rangle - \langle v_i, v_{j} \rangle\langle v_i, v_{j'} \rangle  \Big) \\
			&\le& \frac{1}{\Big(\frac34 + \gamma  - \gamma^2\Big)}\Big(-1/2 + \gamma - \Big(\frac12 - \gamma\Big)^2\Big) \\
			&\le& - 1 + 4 \gamma
		\end{eqnarray*}
		Since the above holds for any pair of vertices $j,j' \in [r]$ which forms an edge, the claim follows.
\end{proof}
The next lemma says that approximately vector $2$-colorable graphs cannot contain short odd cycles. 
\begin{lemma}					\label{lem:appx-2-col(b)}
	Let $G = (V,E)$ be a $(2,\gamma)$-vector colorable, where $\gamma \le 1/16$. Then $G$ does not contain odd cycles of length at most $1/8\sqrt{\gamma}$.
\end{lemma}
\begin{proof}
		Let $v_1,v_2,\ldots,v_n$ be the $(2,\gamma)$-vector coloring of $G$. For contradiction, let $C$ be an odd cycle in $G$ of length $r \le 1/(8\sqrt{\gamma})$. Without loss of generality, let $C = \{1,2,\ldots,r\}$, such that for every $i \in [r]$, the pair $\{i,({i\mod r}) + 1\}$ forms an edge. Let $r = 2k + 1$. Now for any $i \in [r]$, we have $-1 \le \langle v_i,v_{i+1} \rangle \le - 1 + \gamma$. Since $v_i, v_{i+1}$ are unit vectors, we have
		\begin{equation}
			\|v_i + v_{i+1}\|^2 = \|v_i\|^2 + \|v_{i+1}\|^2 + 2 \langle v_i, v_{i+1} \rangle \le 2\gamma
		\end{equation}
		which implies that $\|v_i + v_{i+1}\| \le 2\sqrt{\gamma}$ i.e, any consecutive pair of vectors are {\em almost anti-podal}. Then, for any $i \in [r]$ we also get that 
		\begin{equation}			\label{eq:vec-bound}
			\|v_i - v_{i+2}\| \le \|v_i + v_{i+1}\| + \|v_{i+1} + v_{i+2}\| \le 4\sqrt{\gamma}
		\end{equation}
		We shall now use the above observations to arrive at a contradiction. From the upper bound on $r$, we have $k \leq (r-1)/2 \leq 1/(16\sqrt{\gamma})$, and hence using Eq. \ref{eq:vec-bound} we get that 
		\begin{eqnarray}
			\|v_1 - v_r\| \le \sum_{j = 0}^{k-1} \|v_{1 + 2j} - v_{1 + 2(j+1) } \| \le 4k\sqrt{\gamma} < 1/4
		\end{eqnarray}
		But on the other hand, since $v_1,v_r$ are consecutive vertices in the cycles $C$, we also have $\langle v_1 ,v_r \rangle \le -1 + \gamma$ which implies that $\|v_1 - v_r\| \ge \sqrt{4 - 4 \gamma} > 1$, which give us the contradiction.
\end{proof}

\paragraph{Coloring graphs without short odd cycles}

A key combinatorial tool used in our paper is the following Ramsey theoretic result which says that graphs without short odd cycles contain large independent sets which can be found efficiently. 

\begin{lemma}\cite{RS85ramsey}				\label{lem:ramsey}
	There exists a constant $\epsilon_0 \in (0,1)$ such that for every choice of $0 < \epsilon < \epsilon_0$ the following holds. Let $G = (V,E)$ be a graph without odd cycles of length at most $1/\epsilon$. Then, $G$ contains an independent set of size at least $|V|^{1 - 2\epsilon}$. Furthermore, there exists a polynomial time algorithm which finds such an independent set.
\end{lemma}

Consequently, given a graph without short odd cycles, one can color it efficiently using a small number of colors, as stated in the following corollary.
\begin{corollary}					\label{corr:coloring}
	There exists a constant $\epsilon_0 \in (0,1)$ for which the following holds. Given a graph $G = (V,E)$ which does not contain odd cycles of length at most $1/\epsilon$ where $\epsilon < \epsilon_0$, there exists a polynomial time algorithm which can compute a coloring of $G$ using $\tilde{O}(n^{2\epsilon})$ colors.
\end{corollary} 
Establishing the above corollary using Lemma \ref{lem:ramsey} is straightforward, and just uses the fact that one can keep removing large independent sets in the graph using Lemma \ref{lem:ramsey}, and recurse on the remaining vertices. For the sake of completeness, we include the proof here.

\begin{proof}

\begin{algorithm}						
	\label{alg:ind-set}
	\SetAlgoLined
	\KwIn{Graph $G = (V,E)$}
	Initialize $t \gets 1$ and $G_1 \gets G$\;
	\While{$G_t \neq \phi$}{
		Let $I_t$ be the independent set from Lemma \ref{lem:ramsey} instantiated with $G_t$\;
		Set $G_{t+1} \gets G_t \setminus I_t$\;
		Update $t \gets t + 1$\;  
	}
	
	Output coloring $I_1 \uplus I_2 \uplus \cdots \uplus I_t$\;
	\caption{IndSetColoring}
\end{algorithm}

Consider Algorithm IndSetColoring for coloring by iteratively finding large independent sets. Here, we use Lemma \ref{lem:ramsey} to iteratively remove independent sets $I_1,I_2,\ldots,I_t$, where each independent set forms a color class. Let $G_t = G[V \setminus (I_1 \cup I_2 \cup \cdots I_t)]$ denote the graph on the surviving vertices after $t$ iterations. We claim that in every $T = n^{2\epsilon}$ applications of Lemma \ref{lem:ramsey} at least a constant fraction of vertices are removed, i.e., for any iteration $t$, we have $|{\rm Vert}(G_{t + T})| \leq ( 1- 1/2^{1-2\epsilon})|{\rm Vert}(G_t)|$.

This can be shown as follows. Let $n_t = |{\rm Vert}(G_t)|$ denote the number of vertices in graph $G_t$. Then, we can assume that $|{\rm vert}(G_{t + T})| > n_t/2$ (otherwise we are done). Then, in $T$ iterations the number of vertices removed can be lower bounded by 
\begin{equation}
\sum_{j = 1}^{T}|I_{j+T}|\geq \sum_{j = 1}^{T}|{\rm Vert}(G_{t + j})|^{1 - 2\epsilon} \geq n^{2\epsilon}(n_t/2)^{1 - 2\epsilon} \ge n_t/2^{(1 - 2\epsilon)}
\end{equation}
where the first inequality follows from the guarantee of Lemma \ref{lem:ramsey}. Therefore, in $\tilde{O}(n^{2\epsilon}) $ iterations, all the vertices will be accounted for.

\end{proof}

\section{Approximation algorithm for General Setting}

In this section, we prove our approximation guarantees in the adversarial model, as formally stated in the following theorem:

\ColMain*
	
	\begin{algorithm}[h!]					
		\SetAlgoLined
		Set $\Delta = n^{3/4}$\;
		Solve the Partial-$3$-Coloring SDP (SDP-P$3$C): 
			\begin{alignat*}{4}
					\mbox{minimize }  &  \sum_{i\in V} w_i  \\
					\mbox{subject to } &  \langle v_i, v_j \rangle \leq -\frac{1}{2}+\frac{3}{2}z_{ij} &\qquad & \forall \{i,j\}\in E \\
					& z_{ij} \leq w_i+w_j &\qquad &  \forall \{i,j\}\in E \\
					& 0 \le z_{ij}\leq 1 &\qquad & \forall \{i,j\}\in E \\
					& 0 \le w_{i} \leq 1 &\qquad & \forall i\in V \\
					& \|v_i\|^2=1 &\qquad & \forall  i \in V
			\end{alignat*}\\
			\nonl
		{\it(i) Thresholding}: \\
		Let $S \gets \left\{i \in V| w_i \ge \gamma/3 \right\}$\;
		Let $G' \gets G[V \setminus S]$ be the graph obtained after deleting $S$\;
		\nonl{\it(ii) Coloring Large Degree vertices}:\\
		\While{$\exists i \in G'$ such that ${\rm deg}_{G'}(i) \ge \Delta$}
		{
			Color $G'[\{i\} \cup N_{G'}(i)]$ using $\tilde{O}(n^{C\sqrt{\gamma}})$ colors using the algorithm guaranteed by Corollary \ref{corr:coloring}\;
			Remove $\{i\} \cup N_{G'}(i)$ from $G'$\;
		}
		\nonl{\it(iii) Coloring Low Degree vertices}: \\
	Use {\em randomized rounding} from Theorem \ref{thm:appx-col} to color the remaining vertices in $G'$\;
	\caption{Partial-$3$-Coloring}    \label{alg:3col-main}
	\end{algorithm}

The algorithm for the above theorem is described in Algorithm \ref{alg:3col-main}. In the following subsections, we prove the correctness of the above algorithm. The proof of Theorem \ref{thm:3col-main} can broken down into the analysis of steps (i),(ii) and (iii) of the Partial-$3$-Coloring algorithm. Broadly, we show the following: In step (i), we show that the optimal of the SDP-P$3$C is small (i.e., at most $\epsilon n$), therefore by averaging, the fraction of large $w$ vertices is small. Furthermore, the graph induced on the surviving vertices must satisfy the edge constraints from the SDP with small slack $\gamma$, and therefore must be approximately vector $3$-colorable. As is usual in coloring algorithms, we first iteratively color large degree (i.e., $\ge \Delta$) vertices and their neighborhoods using small number of colors until the graph has degree bounded by $\Delta$ (Claim \ref{cl:step(ii)}). Finally, the remaining graph is also approximately vector $3$-colorable, and has degree bounded by $\Delta$. Therefore, using a hyperplane based randomized rounding procedure to iteratively find large independent sets in $G'$, we can give a $\tilde{O}(\Delta^{1/3 + O(\sqrt{\gamma})})$ coloring of the remaining vertices (Theorem \ref{thm:appx-col}). In the following subsection, we formally prove the steps described above. 

To begin with, we first show that the thresholding step throws away at most a small fraction of vertices.

\begin{claim}[Removing Large Slack Vertices]				\label{cl:step(i)}
	Let $S \subset V$ be as constructed in the thresholding step. Then $|S| \le 3\epsilon n/\gamma$.
\end{claim}
\begin{proof}
	
	We begin by showing that the optimal of SDP-P$3$C is at most $\epsilon n$. Let $V = V_{\rm good} \cup V_{\rm bad}$ be any partition of the vertex sets into good and bad vertices such that (a) $G[V_{\rm good}]$ is $3$-colorable and (b) $|V_{\rm bad}| \le \epsilon n$. Using this partition we now construct a $2$-dimensional feasible solution $(\widehat{v},\widehat{w},\widehat{z})$ to SDP-P$3$C as follows. We set the $\widehat{w}_i$ and $\widehat{z}_{ij}$ variables as  
	\[
	\widehat{w}_{i}= 
	\begin{cases}
	0,  & \text{if } i \in V_{\rm good}\\
	1,    & \text{otherwise}
	\end{cases}
	\qquad\mbox{ and }\qquad \widehat{z}_{ij}= 
	\begin{cases}
	0,  & \text{if } i,j \in V_{\rm good}\\
	1,    & \text{otherwise}
	\end{cases}
	\]
	Furthermore, we set $\{\widehat{v}_i\}_{i \in V_{\rm good}}$ be a vector $3$-coloring of $G[V_{\rm good}]$, and for every $i \in V_{\rm bad}$ we set $\widehat{v}_i = [1 \quad 0]$. We quickly verify that the $\widehat{v},\widehat{w}$ and the $\widehat{z}$ variables constructed as above form a feasible solution to the SDP. By construction, for every $i \in V$ we have $\widehat{w}_i \in [0,1]$ and $\|\widehat{v}_i\|^2 = 1$, and for every edge $(i,j) \in E$ we have $z_{ij} \in [0,1]$. Furthermore, for any edge $(i,j)$ we also have 
	\begin{equation*}
		\widehat{z}_{ij} = \mathbbm{1}\Big(\left\{i \in V_{\rm bad}\right\} \vee \left\{j \in V_{\rm bad}\right\} \Big) \leq \mathbbm{1}\big(\left\{i \in V_{\rm bad}\right\} \big) + \mathbbm{1}\big(\left\{j \in V_{\rm bad}\right\} \big) = \widehat{w}_i + \widehat{w}_j
	\end{equation*}
	All that remains to verify is that the variables also satisfy the approximate vector coloring constraints. We look at two cases: if $i,j \in V_{\rm good}$, then $\widehat{v}_i,\widehat{v}_j$ come from the vector $3$-coloring of $G[V_{\rm good}]$ and therefore they satisfy $\langle \widehat{v}_i, \widehat{v}_j \rangle \leq -\frac12 \leq -\frac12 + \widehat{z}_{ij}$. On the other hand if $i \in V_{\rm bad}$ or $j \in V_{\rm bad}$ then by construction we have $\widehat{z}_{ij} = 1$, and therefore $\langle \widehat{v_i},\widehat{v}_j \rangle \le \|\widehat{v}_i\|\|\widehat{v}_j\| = 1 = -\frac12 + \frac32\widehat{z}_{ij}$.
	
	Therefore, we have established that $(\widehat{z},\widehat{w},\widehat{v})$ are a feasible solution for SDP-P$3$C. Since by construction $\widehat{w}_i = \mathbbm{1}\left\{i \in V_{\rm good} \right\}$, and the $|V_{\rm bad}| \le \epsilon n$, it follows that the SDP optimal $\sum_{i \in V} w_i$ is at most $ \sum_{i \in V} \widehat{w}_i \le \epsilon n$. Therefore, using Markov's inequality, we get
	\begin{equation*}
		|S| = n\cdot \Pr_{i \sim V}\Big[w_i \ge \gamma/3\Big] \le n\cdot\frac{3\sum_{i \in V} w_i}{n\gamma} = \frac{3\epsilon n}{\gamma}
	\end{equation*}
	
\end{proof}

From the above claim, the graph $G' = G[V \setminus S]$ induced on the remaining vertices satisfies the following properties:

\begin{itemize}
	\item[1.] The graph $G'$ contains at least $(1-3\epsilon/\gamma)n$ vertices.
	\item[2.] The graph $G'$ is $(3,\gamma)$-vector colorable. In particular, the vectors $(v_i)_{i \in V \setminus S}$ themselves are a $(3,\gamma)$-vector coloring of $G'$.
\end{itemize}

The second point shall be used crucially in the analysis of the remaining two steps. The next claim bounds the number of colors used while coloring the large degree vertices in step (ii).

\begin{claim}[Degree Reduction]			\label{cl:step(ii)}
	In step (ii), over all the iterations of the while loop, the algorithm uses at most $(n/\Delta)\tilde{O}\left(n^{C\sqrt{\gamma}}\right)$ colors, where $C> 0$ is a constant.
\end{claim}
\begin{proof}
	Fix any vertex $i \in G'$, and let $\tilde{G}_i = G'[N(i)]$ the graph induced on the neighborhood of vertex $i$. Since the graph $G'$ is $(3,\gamma)$-vector colorable, using Lemma \ref{lem:appx-2-col(a)} we know that $\tilde{G}_i$ is $(2, 4\gamma)$-vector colorable. Furthermore, from Lemma \ref{lem:appx-2-col(b)}, we know that $G'$ does not contain odd cycles of length at most $1/(8\sqrt{4\gamma})$. Therefore, we can use Corollary \ref{corr:coloring} to obtain a $\tilde{O}(n^{C\sqrt{\gamma}})$ coloring of $\tilde{G}_i \cup \{i\}$. Finally, note that each iteration of the for loop removes and colors at least $\Delta + 1$ vertices of the graph. Therefore, the total number of iterations of the for loop is bounded by $n/\Delta$. Since in each such iteration we can color the vertex and its neighborhood using $n^{C\sqrt{\gamma}}$ number of colors, the claim follows.
\end{proof}

After steps $(i)$ and $(ii)$, we are left with the graph $G'=(V',E')$ which is $(3,\gamma)$-vector colorable graph and has degree at most $\Delta$. In particular, for every edge $(i,j) \in E'$, the corresponding vectors satisfy $\langle v_i , v_j \rangle \le -\frac12 + \gamma$. Since the independent set based rounding technique~\cite{KMS98}~\cite{AC06} for coloring vector $3$-colorable graphs is \emph{robust}, we can still use it to round the vector coloring of approximately $3$-colorable graphs with similar guarantees, as formally stated in the following theorem.

\begin{theorem}						\label{thm:appx-col}
	Let $G = (V,E)$ be a graph with maximum degree $\Delta$ which is $(3,\alpha)$-vector colorable. Then there exists an efficient randomized algorithm that can color it using $O\left((\ln \Delta)^{1 / 2} \Delta^{\frac{\frac{3}{4}+\alpha-{{\alpha}^2}}{(\frac{3}{2}-\alpha)^2}  }\ln{n}\right)$ colors. 
	
	In particular, if $\alpha \le 1/10$, then the algorithm uses at most $\tilde{O}\left((\ln \Delta)^{1 / 2} \Delta^{{\frac{1}{3}+10\alpha}} \right)$, where $\tilde{O}$ hides polylogarithmic factors in $n$.
\end{theorem}

The proof of the above theorem is an extension of the proofs from \cite{KMS98,AC06} to the setting of approximately vector $3$-colorable graphs. We defer the proof to Appendix \ref{sec:3col-indset}. Instantiating the above theorem with $G = G'$ and $\alpha = \gamma$, we get that $G'$ is colored using $\tilde{O}(\Delta^{1/3 + 10\gamma})$ colors. Overall, the algorithm throws away at most $3 \epsilon/\gamma$ fraction of vertices in step (i). Furthermore, it uses a total of $\tilde{O}\left((n/\Delta)n^{O(\sqrt{\gamma})} + \Delta^{1/3 + 10\gamma}\right)$ colors in steps (ii) and (iii) respectively. Setting $\Delta = n^{3/4}$ in the previous expression, we get that the algorithm uses at most $\tilde{O}(n^{1/4 + O(\sqrt{\gamma})})$ colors. This concludes the analysis of the Partial-$3$-Coloring algorithm and the proof of Theorem \ref{thm:3col-main}.

\section{Algorithm for Semi-random instances}

In this section, we prove Theorem \ref{thm:3col-random}, which we again state here for convenience. 

\ColRandom*

We begin by describing the algorithm for the semi-random setting:

		\begin{algorithm}
			\SetAlgoLined
			Let $\mathcal{A}$ be the algorithm which can color $3$-colorable graphs using $n^\theta$ colors\;
			Set $\delta = \theta/10$\;
										
			\nonl\texttt{\\}		
			\nonl\{{\it Many short odd cycles}\}: \\
			{	
				\For{every vertex $v \in V$}
				{
				Let $G_v := G[N_{G}(v)]$ the subgraph induced by the neighborhood of $G$\;
				Greedily construct a maximal set $\mathcal{C}_v$ of vertex disjoint odd cycles of length at most $1/\delta$ in $G_v$\;	
				}
				Construct set $S \gets \{v \in V : |\mathcal{C}_v| \ge 2\epsilon n \}$\;
				Let $G_0 \gets G[V \setminus S]$ be the graph obtained after deleting $S$\;
				Let $\sigma_1$ be the coloring of $V \setminus S$ obtained by running algorithm $\mathcal{A}$ on $G_0$. Let $L$ denote the number of colors used by the algorithm\;	
			}
			
			\nonl\texttt{\\}
			\nonl\{{\it Few short odd cycles}\}:		\\
			{ 	 	
				Compute a maximal set $\mathcal{C} = \left\{C_1,C_2,\ldots,C_m\right\}$ of vertex disjoint odd cycles in $G$ of length at most $1/\delta$ using greedy algorithm\;
				Let $V' = V \setminus \left( \bigcup_{i \in [m]} {\rm vert}(C_i)\right)$\;
				Use the algorithm guaranteed by Corollary \ref{corr:coloring} to give a $\tilde{O}\left({n^{2\delta}}\right)$ coloring $\sigma_2$ of $G[V']$\;
			}	
			
			\nonl\texttt{\\}
		    \nonl\{{\it Output best coloring}\}:		\\
			\If{$|S| \le \epsilon n$ and $L \le n^\theta$}
			{Output coloring $\sigma_1$ of $V \setminus S$}
			\Else
			{
				Output coloring $\sigma_2$ of $V'$\;
				}
			
		\caption{P$3$C-Random}
		\end{algorithm}

		The algorithm proceeds case wise depending on whether there exists many vertex disjoint short odd cycles in $G$. If it does, then since $V_{\rm bad}$ is small, $G[V_{\rm good}]$ must also contain many vertex disjoint odd cycles. We show that these short cycles will show up in the neighborhood of the bad vertices with high probability, which can be used to identify them. On removing these vertices, we will be left with a $3$-colorable graph. On the other hand, if the number of short odd cycles is small, we can remove them. The remaining graph will still contain most of the vertices and will be short odd cycle free. We can then use Lemma \ref{lem:ramsey} to recover large independent sets. Finally, since the odd cycles we consider are of length at most $1/\delta$, we can work with a \emph{maximal} set of vertex disjoint odd cycles, instead of the largest cardinality set of vertex disjoint odd cycles, while only losing a factor of $1/\delta$ in our analysis.

	\subsection{Correctness of the P$3$C-Random algorithm}
	
	Let $\mathcal{C}^* = \{C^*_1,C^*_2,\ldots,C^*_{m^*}\}$ be a {\em fixed largest cardinality set of vertex disjoint odd cycles} of length at most $1/\delta$ in $G[V_{\rm good }]$. {\em In particular, $\mathcal{C}^*$ and consequently $m^*$, does not depend on the realization of the random and adversarial edges (i.e., the $E_0$ and $E_1$ edges) between $V_{\rm good}$ and $V_{\rm bad}$}. We break our analysis into two cases depending on whether $m^*$ is small or large.

	{\bf Case (i)} $m^* > 4\epsilon n/(\delta p^{1/\delta})$ : For ease of exposition, we say that an odd cycle $C$ in graph $G$ is \emph{good} if it consists of only good vertices, otherwise we call it {\em bad}. The first claim shows the set $\mathcal{C}_v$ must be small for good vertices.
	
	\begin{claim}					\label{cl:3col-a}
		For every good vertex $v \in V$, we have $|\mathcal{C}_v| \le \epsilon n$.
	\end{claim}
	\begin{proof}
		Fix a good vertex $v \in V_{\rm good}$. We claim that a good cycle $C$ can never appear in the neighborhood of a good vertex. For contradiction, let $C$ be a good odd cycle appearing in the neighborhood of $v$. Let $\tilde{G} = G\Big[{\rm vert}(C) \cup \{v\}\Big]$ be the subgraph induced on the vertex $v$ and the vertices from cycle $C$. Since $\tilde{G} \subseteq G[V_{\rm good}]$, the subgraph $\tilde{G}$ is also $3$-colorable. Hence, the neighborhood of $v$ in the induced subgraph $\tilde{G}$ must be $2$-colorable, and therefore it cannot contain odd cycles, and in particular $C$. This gives us the contradiction. 
		
		Hence, any odd cycle which appears in the neighborhood $N_G(v)$ must be bad. Since the number of bad vertices is bounded by $\epsilon n$, and the cycles in $\mathcal{C}_v$ are vertex disjoint, the claim follows.
	\end{proof}
	
	On the other hand, with high probability, we show that $|\mathcal{C}_v|$ is large for all the bad vertices.
	
	\begin{claim}					\label{cl:3col-b}
		With probability at least $1 - e^{-O(\epsilon n)}$, every vertex $v \in V_{\rm bad}$ satisfies $|\mathcal{C}_v| \ge 2\epsilon n$.
	\end{claim}
	\begin{proof}
		Consider the subgraph $G'(V,E_0)$ consisting of edges from  $E_0$ (i.e., the randomly distributed set of edges). Fix a bad vertex $v \in V_{\rm bad}$, and let $G_v = G[N_{G}(v)]$ denote the subgraph induced by the neighborhood of $v$. We shall first give a high probability lower bound on the number of odd cycles from $\mathcal{C}^*$ which can appear in $N_{G}(v)$. Recall that $|\mathcal{C}^*| = m^*$. We also point out again that the choice of $\mathcal{C}^*$ is not affected by the choice of $E_0$ and $E_1$ edges, and can be fixed ahead.

		For every $i \in [m^*]$, we define $Z_i := \mathbbm{1}\Big({\rm vert}(C^*_i) \subseteq N_{G'}(v)\Big)$ to be the indicator random variable that the $i^{th}$ cycle appears in the neighborhood of vertex $v$ in the graph $G'$. Note that these random variables depend only on the realization of the $E_0$ edges. Then we have  
		\begin{eqnarray*}
		\E_{G}[Z_i] \ge \Pr_{E_0}\Big[{\rm vert}(C^*_i) \subseteq N_{G}(v) \Big] &\geq& \Pr_{E_0}\Big[{\rm vert}(C^*_i) \subseteq N_{G'}(v) \Big] \\
		&=& \Pr_{E_0}\Big[\forall j \in {\rm vert}(C^*_i), j \in N_{G'}(v)\Big] \\ 
		&\ge& p^{|C^*_i|} \ge p^{1/\delta}
		\end{eqnarray*}
		Here the last step uses the fact that any cycle $C^*_i \in \mathcal{C^*}$ has length at most $1/\delta$. It follows that 
		\begin{equation}
		\E_G\left[\sum_{i \in [m^*]} Z_i\right] = \sum_{i \in [m^*]}\E_G[Z_i] \ge m^*p^{1/\delta} \ge (4\epsilon/\delta)n
		\end{equation}		
		
		Furthermore, since the cycles $C^*_1,C^*_2,\ldots,C^*_{m^*}$ are vertex disjoint, the corresponding random variables $Z_1,Z_2,\ldots,Z_{m^*}$ are also independent. Therefore using Chernoff bound we get that  
		\begin{equation}					\label{eq:cycle-bound}
		\Pr_G\left[\sum_{i \in [m^*]} Z_i < (2\epsilon/\delta)n \right] \le \Pr_G\left[\sum_{i \in [m^*]} Z_i < \frac12 \E\Big[\sum_{i \in [m^*]} Z_i \Big]\right] \le e^{-\epsilon n/4\delta}
		\end{equation}

		Now let $\mathcal{C}^*_v = \{C^*_i : i \in [m^*], Z_i = 1\}$ be the set of cycles from $\mathcal{C}^*$ which appear in the neighborhood of $v$ in graph $G$ due to the $E_0$ edges. Furthermore,  let $\widetilde{\mathcal{C}}_v$ be a \emph{largest cardinality} set of vertex disjoint odd cycles of length at most $1/\delta$ in $G_v$ (which contains edges from both $E_0$ and $E_1$). Then by definition we have $|\widetilde{\mathcal{C}}_v| \ge |\mathcal{C}^*_v|$. On the other hand, by construction, the set $\mathcal{C}_v$ is a \emph{maximal set} of such vertex disjoint odd cycles in $G_v$, and therefore, it must be a $\delta$-approximation to the largest cardinality set $\widetilde{\mathcal{C}}_v$ i.e.,  $|\mathcal{C}_v| \ge \delta|\widetilde{\mathcal{C}}_v|$ (see Proposition \ref{prop:max-bound}). Therefore using Equation \ref{eq:cycle-bound}, with probability at least $1 - e^{-\epsilon n/4\delta}$ we have
		\begin{equation*}
		|\mathcal{C}_v| \ge \delta|\widetilde{\mathcal{C}}_v| \ge \delta|\mathcal{C}^*_v| \ge 2\epsilon n
		\end{equation*}
		
		Hence, for any fixed vertex $v \in V_{\rm bad}$, w.h.p. we have $|\mathcal{C}_v| \ge 2 \epsilon n$. Therefore, by a union bound and using the lower bound on $\epsilon$, we get that $\Pr_{G}\Big[\exists v \in V_{\rm bad} : |\mathcal{C}_v| < 2\epsilon n\Big] \leq \epsilon n e^{-\epsilon n/4\delta} \leq  e^{-\epsilon n/8\delta}$.
		
	\end{proof}
		 
	Combining the two claims above, it follows that w.h.p. the set $(V\setminus S)$ must exactly be the set of good vertices, and therefore $G[V\setminus S]$ must be $3$-colorable. Hence  algorithm $\mathcal{A}$ will give a $n^{\theta}$ coloring of $G[V \setminus S]$.
	
	{\bf Case (ii)} $m^* \le 4\epsilon n/(\delta p^{1/\delta})$. Let $\mathcal{C} = \mathcal{C}_{\rm good} \uplus \mathcal{C}_{\rm bad}$ be the partition of $\mathcal{C}$ into the set of good and bad cycles respectively. Then, since $\mathcal{C}_{\rm good}$ is a set of vertex disjoint odd cycles of length at most $1/\delta$ in $G[V_{\rm good}]$, it follows that $|\mathcal{C}_{\rm good}| \le |\mathcal{C^*}| \le 4\epsilon n/(\delta p^{1/\delta})$. Furthermore, by arguments similar to the proof of Claim \ref{cl:3col-a}, we have $|\mathcal{C}_{\rm bad}| \le \epsilon n$. Therefore, combining the two bounds, we have $|\mathcal{C}| \leq 5\epsilon n/(\delta p^{1/\delta})$. Since every cycle $C \in \mathcal{C}$ contains at most $1/\delta$ vertices, the total number of vertices thrown away at this step is at most $5\epsilon n/(\delta^2 p^{1/\delta})$. Furthermore, using the {\em maximality} of $\mathcal{C}$, we know that the induced subgraph $G' = G[V']$ must be free of odd cycles of length at most $1/\delta$. Therefore, using Corollary \ref{corr:coloring}, we can color $G'$ using $\tilde{O}(n^{2\delta})$ colors. This concludes the analysis of case (ii).
	
	{\bf Putting Things Together}: If case (i) holds, then w.h.p., in the {\em Many short odd cycles} block of the algorithm, the set $S$ constructed is identical to $V_{\rm bad}$, in which case the algorithm $\mathcal{A}$ will find a $n^\theta$-coloring of $G[V\setminus S] = G[V_{\rm good}]$. In particular, this implies that the conditions of the "if" block will be satisfied and the algorithm will return a $n^\theta$-coloring of $(1-\epsilon)n$ vertices. 
	
	On the other hand, if case (ii) holds, we know that $m \leq 5\epsilon n/(p^{1/\delta}\delta)$, and the {\em Few short odd cycles} block deletes at most $5\epsilon n/(p^{1/\delta}\delta^2)$ vertices, and colors the remaining vertices using $\tilde{O}(n^{2\delta})$ colors. Then the {\em else} block of the algorithm will return a $\tilde{O}(n^{2\delta})$ coloring of $\Big(1 - 5\epsilon n/(p^{1/\delta}\delta^2)\Big)n$ vertices. Since the else block is evaluated only when the conditions of the {\em if} block are not satisfied, it follows that in this case, the algorithm will throw away at most $\max\left(\epsilon n, 5\epsilon n/(\delta^2p^{1/\delta})\right) = O(\epsilon n /\delta^2p^{1/\delta})$ vertices, and color the remaining graph with at most $\max\left(n^{\theta},\tilde{O}(n^{2\delta})\right) = n^\theta$ colors. 
	
	Combining the two above cases gives us Theorem \ref{thm:3col-random}.

\section{Partial $2$-Coloring in the Semi-random model}

In this section, we give an efficient approximation algorithm for partial $2$-coloring problem in the semi-random model with tighter guarantees. The following theorem formally states our guarantees for this setting.

\TwoColRandom*

The algorithm for the above theorem (described as Algorithm \ref{alg:alg-2col-random}) is quite similar to P$3$C-Random algorithm, but overall, the algorithm and its analysis are much simpler. We begin by describing the algorithm.

		\begin{algorithm} 
			\SetAlgoLined
			For every vertex $v \in V$, compute a greedy triangle count as follows: \\
			\For{$v \in V$}
			{
			Let $G_v = G[N_G(v)]$ be the graph induced on the neighborhood of $v$\;
			Construct a maximal matching $T(v)$ in $G_v$ using greedy algorithm\;
			Set $t(v) \gets |T(v)|$\;
			}
			Let $S \gets \{v \in V: t(v) \ge 2\epsilon n\}$\;
			Let $G_0 = G[V\setminus S]$\;
			Let $G_1 \subseteq G$ be the independent set obtained using the $2$-factor approximation for Vertex Cover on $G$\;
			\If{$|{\rm vert}(G_0)| \geq |{\rm vert}(G_1)|$ and $G_0$ is bipartite}
			{
				Output bipartite graph $G_0$\;
			}
			\Else{
				Output independent set $G_1$\;
			}	
			
		\caption{P$2$C-Random}		\label{alg:alg-2col-random}
		\end{algorithm}

The key difference here is that the algorithm uses triangles as forbidden subgraphs for identifying bad vertices instead of neighborhoods with short odd cycles. As before, the algorithm broadly addresses two cases depending on the size of the maximum matching in $G[V_{\rm good}]$. Suppose the subgraph $G[V_{\rm good}]$ contains a linear sized matching $M$. Then, for every bad vertex $v \in V_{\rm bad}$, with high probability, at least one of the matching edges from $M$ will appear in the neighborhood of $v$, which together will form a triangle, which can then be used to identify the bad vertices. On the other hand, if the size of maximum matching in $G[V_{\rm good}]$ is small, then the subgraph $G[V_{\rm good}]$ and consequently $G$ must admit a small sized vertex cover. Therefore, using the greedy approximation algorithm for vertex cover, we can find a small sized vertex cover, whose complement must be a large independent set (which is $1$-colorable).

\subsection{Proof of Theorem \ref{thm:2col-random}}

Let $M \subseteq G[V_{\rm good}]$ be a {\em fixed matching of maximum size} in $G[V_{\rm good}]$, and let $m^* := |M|$ denote the size of the maximum matching. {\em We point out that the matching $M^*$ is not affected by the realization of edges between $V_{\rm good}$ and $V_{\rm bad}$ (i.e, the $E_0$ and $E_1$ edges)}. As before, we break the analysis into two cases depending on whether $m^*$ is small or large.

{\bf Case (i): $m^* \ge (8\epsilon/p^2) n$}: This case is similar to case (i) of the proof of Theorem \ref{thm:3col-random}. We begin by stating and proving two lemmas which say that the greedy triangle count $t(v)$ is small for all the good vertices, and large for all the bad vertices.

\begin{lemma}						\label{lem:2col-(a)}
	For every good vertex $v \in V_{\rm good}$, we have $t(v) \leq \epsilon n$ 
\end{lemma}
\begin{proof}
	Fix a good vertex $v \in V_{\rm good}$, and let $T(v)$ be a set of edges as constructed in the algorithm. Observe that every edge $(a,b) \in T(v)$ along with vertex $v$ induces a triangle in $G$. Furthermore, since $G[V_{\rm good}]$ is bipartite (and hence triangle free), any triangle $T \subseteq G$ must contain at least one bad vertex. Therefore, as the vertex $v$ is good, every edge $e \in T(v)$ must contain at least one bad vertex. Finally, we observe that the edges in $T(v)$ are vertex disjoint, and there are at most $\epsilon n$ bad vertices, which together implies that $t(v) = |T(v)| \leq \epsilon n$.
\end{proof}

\begin{lemma}						\label{lem:2col-(b)}
	With probability at least $1 - e^{-O(\epsilon n)}$, for every vertex $v \in V_{\rm bad}$, we have $t(v) \ge 2\epsilon n$.
\end{lemma}
\begin{proof}
	Let $G'$ be the subgraph on $G$ consisting of edges from $E_0$ (i.e,. the randomly sampled set of edges). Recall that $M  = \{(a_i,b_i)\}_{i \in [m^*]}\subseteq G[V_{\rm good}]$ is the fixed maximum matching in $G[V_{\rm good}]$ of size $m^*$. Let $Z_i := \mathbbm{1}\Big(\{a_i,b_i \in N_{G'}(v)\}\Big)$ be the indicator variable for the event that $a_i,b_i$ are neighbors of $v$ in the graph $G'$. Then, 
	\begin{equation}
	\E_{G'}\left[\sum_{i \in [m^*]}Z_i\right] = \sum_{i \in [m^*]} \Pr_{G'}\left[\{a_i,b_i \in N_{G'}(v)\}\right] = m^*p^2 \ge 8\epsilon n
	\end{equation} 
	Furthermore, since the edges in $M$ are vertex disjoint, the random variables $Z_1,\ldots,Z_{m^*}$ are independent and identical. Therefore using Chernoff bound we get
	\begin{equation}					\label{eq:match-size}
	\Pr_{G'}\left[\sum_{ i \in [m^*]} Z_i \le 4\epsilon n\right] \le \Pr_{G'}\left[\sum_{ i \in [m^*]} Z_i \le \frac12 \E \sum_{i \in [m^*]} Z_i \right] \le e^{-O(\epsilon n)} 	
	\end{equation} 	
	Let $M_v = \{(a_i,b_i): i\in [m^*],Z_i = 1\}$ be the set of matching edges from $M^*$ appearing in the neighborhood of $v$ in the graph $G'$. Furthermore, let $\tilde{M}_v$ be a maximum matching in the subgraph $G_V: = G[N_G(v)]$ induced on the neighborhood of $v$ (which contains both $E_0$ and $E_1$ edges). Then, by definition we have $|\tilde{M}_v| \ge |M_v|$. On the other hand, by construction, the set $T(v)$ is a maximal matching in the induced subgraph $G_v$. Since a maximal matching is a $2$-approximation to the maximum matching, it follows that $|T(v)| \ge  |\tilde{M}_v|/2 \ge  |M_v|/2 \ge 2\epsilon n$.   
	
	Therefore, for a fixed bad vertex $v \in V_{\rm bad}$, with probability at least $1 - e^{-O(\epsilon n)}$, we have $t(v) \ge 2\epsilon n$. The claim now follows by taking a union bound over all vertices $v \in V_{\rm bad}$.
\end{proof}

Therefore, combining Lemmas \ref{lem:2col-(a)} and \ref{lem:2col-(b)}, we know that with probability at least $1 - e^{-O(\epsilon n)}$, we have $t(v) \le \epsilon n$ if and only if $v \in V_{\rm good}$. Conditioned on this event, the set $S$ must exactly be the set of bad vertices, in which case $G[V \setminus S] = G[V_{\rm good}]$ is bipartite.

{\bf Case (ii): $m^* \le (8\epsilon/p^2) n$}: Since the size of maximum matching in $G[V_{\rm good}]$ is at most $(8\epsilon/p^2) n$, and $G[V_{\rm good}]$ is bipartite, by K\"onig's theorem (Theorem 2.1.1~\cite{diestel2012graph}), it follows that the minimum vertex cover of $G[V_{\rm good}]$ has size at most $(8\epsilon/p^2) n$. Then $G$ has a vertex cover of size at most $(8\epsilon/p^2) n + \epsilon n \leq (10\epsilon/p^2) n$. Therefore, the greedy approximation algorithm for vertex cover returns a vertex cover $S'$ of size at most $(20\epsilon/p^2)n$, and consequently, $V\setminus S'$ will be an independent set of size at least $(1 - (20\epsilon/p^2))n$.

{\bf Putting things together}: In case (i), the algorithm throws away at most $\epsilon n$ vertices and returns a $2$-colorable graph, with probability at least $1 - e^{-O(\epsilon n)}$. In case (ii), the algorithm throws away at most $O(\epsilon/p^2)n$ vertices, and returns an indpendent set. Combining the two cases gives us the guarantees for Theorem \ref{thm:2col-random}.

\section{Conclusion}

In this work we consider the problem of coloring partial $3$-colorable graphs in adversarial and semi-random settings. In the adversarial setting, we give an efficient approximation algorithm which can color $(1 - O(\epsilon^c))$-fraction of vertices using $\tilde{O}(n^{0.25 + \epsilon^{c'}})$ colors. On the other hand, the best known approximation guarantees for $3$-colorable graphs is $n^{0.199}$~\cite{kawarabayashi_thorup_2017}. An obvious open question here is to achieve analogous approximation bounds for partially $3$-colorable graphs as well. 

One direct way to improve on our approximation bounds in the adversarial setting is through the use of more efficient degree reduction mechanisms as typically done in the exact $3$-coloring setting \cite{blum_karger_1997},\cite{kawarabayashi_thorup_2012,kawarabayashi_thorup_2017} using combinatorial techniques like Blum's coloring tools~\cite{Blum}. However, these tools rely on fragile combinatorial properties present in $3$-colorable graphs (e.g. two vertices whose common neighborhood is not an independent set must have the same color in any legal coloring), and as such, it is not obvious how to extend these techniques to the setting of partially $3$-colorable graphs.

In the semi-random model, we show how any efficient algorithm for exact $3$-coloring that uses $n^{\theta}$ colors can be leveraged to obtain an efficient algorithm in this setting which uses the same number of colors with high probability and also does not remove too many vertices. An obvious next step would be to see if similar results can also be obtained for partially $k$-colorable graphs with $k>3$. Another interesting question would be  to see if one can design efficient approximation algorithms with similar guarantees, where the adversary can also delete the randomly sampled edges.

\section*{Acknowledgements}
AL was supported in part by SERB Award ECR/2017/003296. SG would like to thank Pasin Manurangsi for pointing him to the Odd Cycle Transversal problem.

\bibliographystyle{alpha}
\bibliography{main,related}

\appendix
		
\section{Proof of Theorem \ref{thm:appx-col}}				\label{sec:3col-indset}
	
	The proof of the theorem is adapted from \cite{KMS98,AC06} and the rounding algorithm used is identical to theirs. The proof of the theorem goes through the following lemma which says that there exists a randomized algorithm which can find large sized independent sets in approximately vector $3$-colorable graphs with bounded degree.
	The proof presented here is adapted from the proof for the exact $3$-colorable case in Section 13.2 \cite{williamson2011design}.
	
	\begin{lemma}					\label{lem:3col-indset}
		Let $G = (V,E)$ be a graph on $n$ vertices with maximum degree $\Delta$ which is $(3,\alpha)$-vector colorable, where $\alpha \leq 1/2$. Then there exists an efficient randomized algorithm which finds an independent set of size $ \Omega\left(n(\ln \Delta)^{-1 / 2} \Delta^{-\frac{\frac{3}{4}+\alpha-{{\alpha}^2}}{(\frac{3}{2}-\alpha)^2}  }\right)$ with high probability.
	\end{lemma}	
	
	\begin{proof}
	Let $v_1,v_2,\ldots,v_n \in \mathbbm{R}^d$ be a $(3,\alpha)$-vector coloring of $G$. Consider the following randomized hyperplane rounding procedure for finding a independent set:
	\begin{itemize}
		\item[1.]  Draw a random vector $r \sim N(0,1)^d$ by picking each coordinate independently from a standard Gaussian.
		\item[2.]  Compute the sets ${S(\beta)=\left\{i \in V :  \langle v_{i} , r \rangle \geq \beta\right\}}$ and  ${S^{\prime}(\beta)=\left\{i \in S(\beta) : \forall(i, j) \in E, j \notin S(\beta)\right\}}$.
		\item[3.] Return the set $S'(\beta)$.
	\end{itemize}
	
	Here $\beta = \beta(\alpha,\Delta) > 0 $ is a quantity which depends on $\alpha$ and $\Delta$ which will be fixed later. Clearly, by construction, the above procedure returns an independent set. The rest of the proof will just involve lower bounding the expected size of the set $S'(\beta)$.  
	
	Let $\Phi(\cdot)$ denote the gaussian CDF function, and let $\overline{\Phi}(\cdot) \overset{\rm def}{=} 1 - \Phi(\cdot)$. 
	Firstly, we observe that for any $i \in V,$ the probability that $i \in S(\beta)$ is $\overline{\Phi}(\beta)$, where $\Phi(\cdot)$ is the gaussian CDF function. This implies $E[|S(\beta)|]=n \overline{\Phi}(\beta)$.  Now consider the probability that a vertex $i$ is in $S(\beta)$ but not in $S^{\prime}(\beta)$. Observe that 
	$\operatorname{Pr}\left[i \notin S^{\prime}(\beta) | i \in S(\beta)\right]=\operatorname{Pr}\left[\exists(i, j) \in E : \langle v_{j} , r \rangle \geq \beta | \langle v_{i} , r \rangle \geq \beta\right]$. Now fix a pair of neighbouring vertices $(i,j)$. We know that $v_j$ can we written as
	\begin{equation*}
		v_j=\left(-\frac{1}{2}+\alpha_{ij}\right)v_i+\left(\sqrt{\frac{3}{4}+\alpha_{ij}-{{\alpha}^2_{ij}}}\right)u
	\end{equation*} where $u$ is a unit vector orthogonal to $v_i$ and $0 \le \alpha_{ij} \le \alpha$. Going forward, we shall assume that $\alpha_{ij} = \alpha$; it can be easily verified that the same bound holds when $\alpha_{ij} < \alpha$. Rearranging, $u$ can be written as 
	\begin{equation*}
	u=\frac{(\frac{1}{2}-\alpha)}{\sqrt{\frac{3}{4}+\alpha-{{\alpha}^2}}}v_i+\frac{1}{\sqrt{\frac{3}{4}+\alpha-{{\alpha}^2}}}v_j
	\end{equation*}
	Now, consider the inner product of $u$ and $r$, assuming $\alpha\leq\frac12$
\begin{equation}				\label{eqn:bound1}
	\langle u, r \rangle \geq \frac{\left(\frac{1}{2}-\alpha\right)}{\sqrt{\frac{3}{4}+\alpha-{{\alpha}^2}}}\beta+\frac{1}{\sqrt{\frac{3}{4}+\alpha-{{\alpha}^2}}}\beta= \frac{\left(\frac{3}{2}-\alpha\right)}{\sqrt{\frac{3}{4}+\alpha-{{\alpha}^2}}}\beta
\end{equation}
	
	where we use the fact that conditioned on the event $\{i\in S(\beta)\} \wedge \{j \in S(\beta)\}$, we must have $\langle v_i, r \rangle \ge \beta$ and $\langle v_j, r \rangle \ge \beta$ by definition of the set $S(\beta)$. Also, by our choice of parameters we have $\beta > 0$. The above sequence of observations implies the following: conditioned on the event $i \in \beta$, the probability of the event $j \in \beta$ is upper bounded by the probability of event that Eq. \ref{eqn:bound1} holds. Finally, recall that the maximum degree of the graph is at most $\Delta$. Now we proceed to bound the desired probability as follows: 
	\begin{eqnarray*}
	\operatorname{Pr}\Big[\exists(i, j) \in E : \langle v_{j} , r \rangle \geq \beta | \langle v_{i} , r \rangle \geq \beta \Big]  
	&\leq& \sum_{j :(i, j) \in E} \operatorname{Pr}\Big[ \langle v_{j} , r \rangle  \geq \beta | \langle v_{i} , r \rangle \geq \beta\Big] \\
	&\leq& \sum_{j :(i, j) \in E} \operatorname{Pr}\Big[ \left\{\mbox{Eq. \ref{eqn:bound1} holds} \right\}\Big] \\
	&\leq& \Delta \overline{\Phi}\left(\frac{\frac{3}{2}-\alpha}{\sqrt{\frac{3}{4}+\alpha-{{\alpha}^2}}}\beta\right)
	\end{eqnarray*}
	where the last step follows from the fact that for any unit vector $u$, and a gaussian vector $r \sim N(0,1)^d$, the random variable $\langle u,r \rangle$ is distributed as a gaussian, and the definition of $\overline{\Phi}(\cdot)$. 
	
	Recall that $\Delta$ is the maximum degree of the graph. Observe that if we choose $\beta$ such that $\overline{\Phi}\left(\frac{\frac{3}{2}-\alpha}{\sqrt{\frac{3}{4}+\alpha-{{\alpha}^2}}}\beta\right)\leq\frac{1}{2 \Delta},$ then the probability that $i \notin S^{\prime}(\beta)$ given that $i \in S(\beta)$ is at most $\frac{1}{2}$. This would imply that the expected size of $S^{\prime}(\beta)$ is least half
the expected size of $S(\beta),$ which is $\frac{n}{2} \overline{\Phi}(\beta).$\newline
Now, set $\beta=\sqrt{2\ln{\Delta}}\left(\frac{\sqrt{\frac{3}{4}+\alpha-{{\alpha}^2}}}{\frac{3}{2}-\alpha}\right)$. Next we use the following fact about standard normal distributions.
\begin{fact}[Lemma 13.8 \cite{williamson2011design}]
	For  $x>0$, $\frac{x}{1+x^{2}} \phi(x) \leq \overline{\Phi}(x) \leq \frac{1}{x} \phi(x)$, where $\phi(x) := \frac{1}{\sqrt{2\pi}}e^{-x^2/2}$ is the pdf of the standard gaussian distribution. 
\end{fact}
From the above fact, for our choice of $\beta$, we get the following upper bound.

\begin{equation*}
	\overline{\Phi}\left(\frac{\frac{3}{2}-\alpha}{\sqrt{\frac{3}{4}+\alpha-{{\alpha}^2}}}\beta\right)\leq\frac{{\sqrt{\frac{3}{4}+\alpha-{{\alpha}^2}}}}{\beta\Big(\frac{3}{2}-\alpha\Big)} e^{\frac{-\frac12\left(\frac{3}{2}-\alpha\right)^2\beta^2}{\frac{3}{4}+\alpha-{{\alpha}^2}}} \leq \frac{1}{2 \Delta}\quad (1)
\end{equation*}

On the other hand, since $\beta \geq 1$, we have $\frac{\beta}{1 + \beta^2} \geq \frac{\beta}{2\beta^2} = \frac{1}{2\beta}$, and hence we can lower bound $\overline{\Phi}(\beta)$ as follows
\begin{equation*}		
	\overline{\Phi}(\beta)\geq \frac{1}{2 \beta} \frac{1}{\sqrt{2 \pi}} e^{-(\ln \Delta)\frac{\frac{3}{4}+\alpha-{{\alpha}^2}}{(\frac{3}{2}-\alpha)^2}  }=\Omega\left((\ln \Delta)^{-1 / 2} \Delta^{-\frac{\frac{3}{4}+\alpha-{{\alpha}^2}}{(\frac{3}{2}-\alpha)^2}  }\right)\quad(2)
\end{equation*}

Thus,
\begin{eqnarray*}
	E\left[ | S^{\prime}(\beta) |\right]=\sum_{i \in V} \operatorname{Pr}\Big[i \in S^{\prime}(\beta) \Big| i \in S(\beta)\Big] \operatorname{Pr}\Big[i \in S(\beta)\Big] \geq \frac{n}{2} \Omega\left((\ln \Delta)^{-1 / 2} \Delta^{-\frac{\frac{3}{4}+\alpha-{{\alpha}^2}}{(\frac{3}{2}-\alpha)^2}  }\right).
\end{eqnarray*}

which gives us the desired lower bound on the expected size of the independent set output by the randomized rounding procedure.
	\end{proof}
	
	 Now using the above lemma, we complete the proof of theorem. Consider algorithm ApproxHyperplaneColoring for rounding approximate $3$-vector colorings of graphs, which is analyzed in Claim \ref{cl:color}.
	 
	 \begin{algorithm}						
	 	\label{alg:}
	 	\SetAlgoLined
	 	\KwIn{$(3,\alpha)$-vector coloring $\{v_i\}_{i \in V}$ of graph $G = (V,E)$}
	 	Initialize $t \gets 1$,$G_1 \gets G$,$N =\frac{10}{\rho}\ln n$\;
	 	\While{$G_t \neq \phi$}{
	 		Let $I^1,I^2,\ldots,I^N$ be independent sets returned by $N$ i.i.d invocations of Lemma \ref{lem:3col-indset}\;  
	 		Let $I_t \gets \arg\max_{{I^i : i \in [N]}} |I^i|$, and set $G_{t+1} \gets G_t \setminus I_t$\;
	 		Update $t \gets t + 1$\;  
	 	}
	 	Output coloring $I_1 \uplus I_2 \uplus \cdots \uplus I_t$\;
	 	\caption{ApproxHyperplaneColoring}
	 \end{algorithm}
	 
	 \begin{claim}				\label{cl:color}
	 	For any iteration $t$, we have $|I_t| \geq \rho n|{\rm vert}(G_t)|/2$ with probability at least $1 - 1/n^3$.
	 \end{claim}
	 \begin{proof}
	 	Consider one invocation of Lemma \ref{lem:3col-indset}. For any iteration $t$, let $n_t = \left|{\rm vert}(G_t)\right|$ denote the number of surviving vertices in the graph $G_t$.  Then we have have $\E\big[|I|\big] \geq \rho n_t$ or $\E\big[|\overline{I}|\big] \leq (1 - \rho) n_t$. Therefore, we have
	 \begin{equation*}
	 \Pr\Big[\forall i \in [N], |\overline{I}^i| \leq (1 - \rho/2)n_t \Big] = \left(\Pr\Big[|\overline{I}| \leq (1 - \rho/2)n_t \Big]\right)^N 
	 \overset{1}{\leq} \left(\frac{1 - \rho}{1 - \rho/2}\right)^N  
	 \overset{2}{\leq} (1 - \rho/2)^N \overset{3}{\leq} n^{-3} 
	 \end{equation*}
	 where step $1$ uses Markov's inequality, and step $2$ uses the fact $(1 - \rho/2)^2 \geq 1 - \rho$ and step $3$ follows from our choice of $N$.  
	 \end{proof}
	 Therefore, at iteration $t$, with probability at least $1 - n^{-3}$, the size of the independent set returned is at least $\rho |{\rm Vert}(G_t)|/2$, and therefore, the number of surviving vertices drops by a factor of $(1 - \rho/2)$. Hence, with probability at least $1 - 1/n^2$, in $t^* = {O}\left(\frac1\rho\ln n\right)$ iterations, every vertex will be accounted for i.e., they will be part of some independent set. Since each independent set forms a color class, the total number colors used is $O\left(\frac1\rho\ln n\right) = \tilde{O}\left((\ln \Delta)^{1 / 2} \Delta^{\frac{\frac{3}{4}+\alpha-{{\alpha}^2}}{(\frac{3}{2}-\alpha)^2}  }\right)$ colors.

\section{Partial $2$-coloring in the Adversarial Model}

In this section we prove Proposition \ref{prop:2col-gen}, which we recall here for convenience.

\TwoColGen*

The algorithm for the above proposition (described in Algorithm \ref{alg:2col-main})is basically the same as Algorithm 2, with the following key differences: (i) we solve the SDP for approximate vector $2$-coloring with slack constraints (instead of approximate vector $3$-coloring) and (ii) we can directly use Corollary \ref{corr:coloring} to round the vector solution without going through the degree reduction step. 

\begin{algorithm}[h!]					
	\SetAlgoLined
	Solve the Partial-$2$-Coloring SDP (SDP-P$2$C): 
	\begin{alignat*}{4}
		\mbox{minimize }  &  \sum_{i\in V} w_i  \\
		\mbox{subject to } &  \langle v_i, v_j \rangle \leq -{1} + 2z_{ij} &\qquad & \forall \{i,j\}\in E \\
		& z_{ij} \leq w_i+w_j &\qquad &  \forall \{i,j\}\in E \\
		& 0 \le z_{ij}\leq 1 &\qquad & \forall \{i,j\}\in E \\
		& 0 \le w_{i} \leq 1 &\qquad & \forall i\in V \\
		& \|v_i\|^2=1 &\qquad & \forall  i \in V
	\end{alignat*}
	{\it(i) Thresholding}: \;
	Let $S \gets \left\{i \in V| w_i \ge \gamma/4 \right\}$\;
	Let $G' \gets G[V \setminus S]$ be the graph obtained after deleting $S$\;
	\nonl{\it(ii) Round the approximate vector $2$-coloring}: \\
	Use the algorithm from Corollary \ref{corr:coloring} to color the remaining vertices in $G'$\;
	\caption{Partial-$2$-Coloring}    \label{alg:2col-main}
\end{algorithm}

We give a proof sketch of the correctness of the above algorithm.

\subsection{Proof of Proposition \ref{prop:2col-gen}}

To begin with, the following claim shows that in step (i) we do not throw away too many vertices.

\begin{claim}				\label{cl:2colstep(i)}
	Let $S \subset V$ be the set of vertices constructed in step (i) of the algorithm. Then $|S| \leq 4\epsilon n/\gamma$.
\end{claim}

The proof of the above claim is almost identical to that of Claim \ref{cl:step(i)}, hence we omit it here. Therefore after step $(i)$, the subgraph $G' = G[V\setminus S]$ satisfies the following properties:

\begin{enumerate}
	\item The graph $G'$ contains at least $(1 - 4\epsilon/\gamma)n$ vertices.
	\item The graph $G'$ is again $(2,\gamma)$-vector colorable since the vectors $(v_i)_{i \in V \setminus S}$ themselves give a $(2,\gamma)$-vector coloring of the graph.  
\end{enumerate} 

Now from Lemma \ref{lem:appx-2-col(b)} we know that $G'$ cannot contain odd cycles of length at most $1/8\sqrt{\gamma}$. Hence we can use Corollary \ref{corr:coloring} to color $G'$ using $\tilde{O}(n^{2\gamma})$ colors. This concludes the proof of Proposition \ref{prop:2col-gen}.

\section{Maximal and Maximum Short Odd Cycle sets}

\begin{proposition}		\label{prop:max-bound}
	For any graph $G:=(V,E)$, and parameter $\delta \in (0,1)$ the following holds. Let $\mathcal{C}$ be a maximal set of vertex disjoint odd cycles of length at most $1/\delta$, and let $\tilde{\mathcal{C}}$ be a set of largest cardinality of vertex disjoint odd cycles of length at most $1/\delta$. Then $|\mathcal{C}| \geq \delta|\tilde{\mathcal{C}}|$.
\end{proposition}
\begin{proof}
	Since $\mathcal{C}$ is a maximal set of vertex disjoint odd cycles of length at most $1/\delta$, for every odd cycle $\tilde{C} \in \tilde{\mathcal{C}}$, there exists an odd cycle $C \in \mathcal{C}$ such that $C \cap \tilde{C} \neq \emptyset$ i.e,. $\tilde{C}$ is hit by $C$. Now  we observe that (i) the cycles in $\mathcal{C}$ are vertex disjoint and (ii) each cycle $C \in \mathcal{C}$ has size at most $1/\delta$, it follows that any cycle $C \in \mathcal{C}$ hits at most $1/\delta$ cycles in $\tilde{\mathcal{C}}$. Since every cycle in $\tilde{\mathcal{C}}$ is hit by some cycle in $\mathcal{C}$, we must have $|\mathcal{C}| \geq \frac{|\tilde{\mathcal{C}}|}{1/\delta} = \delta|\tilde{\mathcal{C}}|$.
\end{proof}

\section{Identifying the set of Good Vertices is NP-hard}

\begin{fact}				\label{fact1}
For all $k \in \mathbbm{N}$, given a graph $\alpha$-partially $k$-colorable graph $G = (V,E)$ it is NP-Hard to identify a set $V_{\rm good} \subset V$ of size at least $\alpha n$ such that $G[V_{\rm good}]$ is $k$-colorable 
\end{fact}
\begin{proof}
		For $\alpha = 1 - 1/2n$, this is exactly the $k$-Coloring problem which is NP-Hard \cite{Karp72}.
\end{proof}

\end{document}